\documentclass[11pt,english]{article}
\usepackage[margin=1in]{geometry}

\usepackage{amsfonts}
\usepackage{amsmath}
\usepackage{amsthm}
\usepackage{enumerate}
\usepackage{boxedminipage}
\usepackage[square]{natbib}
\usepackage{authblk}

\DeclareMathOperator*{\argmax}{arg-max}

\newcounter{thm}
\newcounter{ex}
\newcounter{re}
\newcounter{def}
\newcounter{pro}
\newtheorem{theorem}[thm]{Theorem}
\newtheorem{lemma}[thm]{Lemma}

\newtheorem{proposition}[pro]{Proposition}

\newtheorem{example}[ex]{Example}

\newtheorem{corollary}[thm]{Corollary}
\newtheorem{defn}[def]{Definition}

\begin{document}

% Full title of the paper (Capitalized)
\title{Bargaining Mechanisms for One-Way Games
\footnote{An earlier, shorter version of this paper appeared in Proceedings of the Twenty-Fourth International joint conference on Artificial Intelligence (IJCAI), \citeyear{abeliuk2015Bargaining}.}}

% Authors (Add full first names)
\author[1,2]{Andr\'es Abeliuk}
\author[1,3]{Gerardo Berbeglia}
\author[1,4]{Pascal Van Hentenryck}

%\author{Andr\'es Abeliuk $^{1,*}$, Gerardo Berbeglia $^{2}$ and Pascal Van Hentenryck $^{3}$}

\renewcommand\Authands{ and }

% Affiliations / Addresses 
\affil[1]{National ICT Australia (NICTA)}
\affil[2]{University of Melbourne, Australia}
\affil[3]{Melbourne Business School, University of Melbourne, Australia}
\affil[4]{Australian National University, Australia}

% Contact information of the corresponding author (Add [2] after \corres if there are more than one corresponding author.)
%\corres{andres.abeliuk@nicta.com.au}

\date{}
\maketitle

\begin{abstract}
We introduce one-way games, a framework motivated by applications in
  large-scale power restoration, humanitarian logistics, and
  integrated supply-chains. The distinguishable feature of the games is that the
   payoff of some player is determined only by her own strategy and
  does not depend on actions taken by other players.
   We show that the equilibrium outcome in
  one-way games without payments and the social cost of any ex-post
  efficient mechanism, can be far from the optimum. We also show that
  it is impossible to design a Bayes-Nash incentive-compatible
  mechanism for one-way games that is budget-balanced, individually
  rational, and efficient. To address this negative result, we propose a privacy-preserving
  mechanism that is incentive-compatible and budget-balanced,
  satisfies ex-post individual rationality conditions, and produces an
  outcome which is more efficient than the equilibrium without payments.
  The mechanism is based on a single-offer bargaining and we show that a randomized
  multi-offer extension brings no additional benefit.
\end{abstract}

\section{Introduction}

When modeling economic interactions between agents, it is standard to
adopt a general framework where payoffs of individuals are dependent
on the actions of all other decision-makers. However, some agents may
have payoffs that depend only on their own actions, not on actions
taken by other agents. In this paper, we explore the consequences of
such asymmetries among agents. Since these features lead to a
restricted version of the general model, the hope is that we can
identify mechanisms that produce efficient outcomes by exploiting the
properties of this specific setting.

A classic application of this setting is Coase's example of a polluter
and a single victim, e.g., a steel mill that affects a laundry. The
Coase theorem (\citeyear{coase1960problem}) is often interpreted as a
demonstration of why private negotiations between polluters and
victims can yield efficient levels of pollution without government
interference. However, in an influential article, Hahnel and Sheeran
(\citeyear{hahnel2009misinterpreting}) criticize the Coase theorem by
showing that, under more realistic conditions, it is unlikely that an
efficient outcome will be reached.  They emphasize that the solution
is a negotiation, and not a market-based transaction as described by
Coase. As such, incomplete information plays an important role and
game theory and bargaining games can explain inefficient outcomes.

Other real-life applications are found in large-scale restoration of
interdependent infrastructures after significant disruptions
\citep{cavdaroglu2013integrating,coffrin2012last}, humanitarian
logistics over multiple states or regions \citep{van2010strategic},
supply chain coordination (see, e.g., Voigt (\citeyear{Voigt2011})),
integrated logistics, and the joint planning and the control of gas
and electricity networks. Consider for example the restoration of
the power system and the telecommunication network after a major
disaster.  As explained in \citet{cavdaroglu2013integrating}, there are one-way
dependencies between the power system and the telecommunication
network. This means, for instance, that some power lines must be
restored before some parts of the telecommunication network become
available.  It is possible to use centralized mechanisms for restoring
the system as a whole.  However, it is often the case that these
restorations are performed by different agencies with independent
objectives and selfish behavior may have a strong impact on the social
welfare. It is thus important to study whether it is possible to find
high-quality outcomes in decentralized settings when stakeholders
proceed independently and do not share complete information about
their utilities.

This paper aims at taking a first step in this direction by proposing
a class of two players one-way dependent decision settings which
abstracts some of the salient features of these applications and
formalizes many of Hahnel and Sheeran's critiques.  We present a
number of negative and positive results on one-way games. We first
show that Nash equilibria in one-way games, under no side payments,
can be arbitrarily far from the optimal social welfare.
Moreover, in contrast to Coase theorem, we show that when side payments are allowed
in a Bayes-Nash incentive-compatible setting, there is no ex-post efficient
individually rational, and budget-balanced mechanism for one-way
games.
To address this negative result, we focus on mechanisms that are
budget-balanced, individually rational, incentive-compatible that are relatively efficient. Our main positive result is a single-offer bargaining
mechanism which under reasonable assumptions on the players,
increases the social welfare compared to the setting where no side payments are allowed.
We also show that this single-offer mechanism cannot be improved by a (randomized) multi-offer
mechanism.

The rest of this paper is organized as follows. In Section \ref{section-game} we define
one-way games and study the properties of Nash equilibria. In section \ref{impossibility_result_section} we prove an impossibility result. The single offer mechanism is presented and analyzed in  sections \ref{section_single_offer} and \ref{section_price_of_anarchy} respectively. Finally, in section \ref{section_multi_offer} we present a multi-offer mechanism and we show that it doesn't improve the efficiency with respect to the single-offer one.

\section{One-Way Games}
\label{section-game}

\emph{One-way games} feature two players $A$ and $B$. Each player
$i\in{A,B}$ has a public action set $\mathcal{S}_i$ and we write
$\mathcal{S}=\mathcal{S}_A \times \mathcal{S}_B$ to denote the set of
joint action profiles. As most commonly done in mechanism design, we
model private information by associating each agent $i$ with a payoff
function $u_i:\mathcal{S} \times \Theta_i \rightarrow \mathbb{R}^+$,
where $u_i(s, \theta_i)$ is the agent utility for strategy profile $s$ when the
agent has type $\theta_i$. We assume that the player types are
stochastically independent and drawn from a distribution $f$ that is
common knowledge. We denote by $\Theta_i$ the possible types of player
$i$ and write $\Theta= \Theta_A \times \Theta_B$. If $\theta \in
\Theta$, we use $\theta_i$ to denote the type of player $i$ in
$\theta$. Similar conventions are used for strategies, utilities, and
type distributions.

A key feature of one-way games is that the payoff $u_A((s_A,s_B), \theta_A)$ of
player $A$ is determined only by her own strategy and does not depend on
$B$'s actions, i.e.,	
\[
\forall s_A, s_B, s_B',\theta_A: u_A((s_A,s_B),\theta_A) =
u_A((s_A,s_B'),\theta_A).
\]
As a result, for ease of notation, we use $u_A(s_A,\theta_A)$ to
denote $A$'s payoff. Obviously, player $B$ must act according to what
player $A$ chooses to do and we use $s_B(s_A,\theta_B)$ to denote the
best response of player $B$ given that $A$ plays action $s_A$ and
player $B$ has type $\theta_B$, i.e.,
\[
s_B(s_A,\theta_B) = \argmax_{s_B \in \mathcal{S}_B}u_B((s_A,s_B), \theta_B)
\]
where ties are broken arbitrarily. In this paper, we always assume
that ties are broken arbitrarily in $\argmax$ expressions.

One-way games assumes that players are risk-neutral agents and that after having observed the
realization of their own types, players simultaneously choose their actions.
As a consequence, if side payments are not allowed, player $A$ will play an action $s^N_A$ that yields her a
maximum payoff, i.e.,
\[
s^N_A(\theta_A) = \argmax_{s_A \in \mathcal{S}_A}u_A(s_A,\theta_A),
\]
Player $B$ will pick $s^N_B(\theta_B)$ such that her expected payoff
is maximized, i.e.,
\[
s_B^N(\theta_B) =  \argmax_{s_B \in \mathcal{S}_B} \mathbb{E}_{\theta_A} \left[ u_B(s_B(s^N_A(\theta_A),\theta_B),\theta_B)\right].
\]

\begin{sloppypar}
\noindent
The set of Nash equilibria (NE) is thus characterized by
$s^N(\theta)=(s^N_A(\theta_A),s^N_B(\theta_B)) \subseteq \mathcal{S}$.  The
best response $s_B^N(\theta_B)$ of player $B$ may be a bad outcome for
her even when $B$ has a much greater potential payoff.  Player $A$
achieves its optimal payoff, but our motivating applications aim at
optimizing a global welfare function
\[
SW((s_A,s_B),\theta) = u_A(s_A,\theta_A) + u_B((s_A,s_B),\theta_B).
\]
Thus, the global welfare achieved by the Nash equilibria can be expressed as  $SW(s^N(\theta),\theta)$.
We quantify the quality of the Nash equilibrium outcome with the price
of anarchy (PoA).
\end{sloppypar}

\begin{defn}
The price of anarchy of $s^N(\theta)\subseteq \mathcal{S}$ given type $\theta \in \Theta$ is defined as
\[ PoA(\theta) = \frac{\max_{s\in\mathcal{S}} SW(s,\theta)
}{\min_{s\in s^N(\theta)}SW(s,\theta)}. \]
\end{defn}
\noindent
A natural extension for the price of anarchy is to quantify the
expected worst-case equilibrium.

\begin{defn}
The Bayes-Nash price of anarchy is defined as
\[
PoA = \mathbb{E}_{\theta} \left[ PoA(\theta) \right].
\]
\end{defn}
\noindent
Note that the price of anarchy given type $\theta \in \Theta$ can be used to obtain a lower and an upper
bound on the Bayes-Nash price of anarchy in the following way,
\[
\min_{\theta \in \Theta} PoA(\theta) \leq PoA \leq \max_{\theta \in \Theta} PoA(\theta).
\]

\noindent
Throughout this paper, we use the following two running examples to illustrate key concepts.
\begin{example}\label{ex:1b}
  \normalfont Consider the instance where player $A$ has two possible
  actions $s_A^1, s_A^2 \in \mathcal{S}_A$.  Action $s_A ^1$ has a
  payoff  $u_A(s_A^1)$ distributed according to a uniform distribution between $0$
  and $100$, while action $s_A^2$ has a constant payoff  $u_A(s_A^2)=100$. The
  set of dominant actions for player B corresponds to the set of best
  responses $s_B(s_A^1)$ and $s_B(s_A^2)$. Let player $B$ have only one type and
  we set payoffs to be $u_B(s_B(s_A^1))=x$ and $u_B(s_B(s_A^2))=0$, where $x$ is a positive
  constant.  When no transfers are allowed, player $A$ will always
  play action $s_A^2$, yielding a social welfare of $100+0$.  If
  player A plays $s_A^1$, her expected payoff is $50$ and the expected
  social welfare is thus $50+x$. The price of anarchy is
  $\frac{50+x}{100}$ if $x\geq 50$ and $1$ otherwise. Notice that the
  PoA is an increasing function of $x$.
\end{example}
\begin{example}\label{ex:1}
  \normalfont
  \begin{sloppypar}
   Consider the instance where player $A$ has $n$ possible
  actions $s_A^1, s_A^2,\ldots, s_A^n \in \mathcal{S}_A$.  Each action has a
  payoff $u_A(s_A^i)$ which is independent and identically distributed according to
  a Uniform distribution between $0$ and $1$.
  For player $B$,  consider the set of best responses
  $s_B(s_A^1,\theta_B), s_B(s_A^2,\theta_B), \ldots, s_B(s_A^n,\theta_B)$.
  Let the expected payoffs be  $\mathbb{E}_{\theta_B}\left[u_B(s_B(s_A^i,\theta_B),\theta_B)\right]=\mu_i$ for $i=1,2,\ldots,n$,
  where $\mu_1\geq\mu_2\geq \cdots \geq \mu_n$.  All other payoffs are
  set to $0$, i.e., for any $\theta_B \in \Theta_B$, $s_A \in \mathcal{S}_A$ and $s_B\not =  s_B(s_A,\theta_B)$: $u_B\left( (s_B,s_A),\theta_B \right)=0$.
\end{sloppypar}
   When no transfers are allowed, player $A$ chooses her (realized) maximizing payoff
  action. Such action is distributed according to the largest order statistic, i.e.,
  the maximum between all payoffs. The largest order statistic between $n$
  standard Uniforms follows a $Beta(n,1)$ distribution, with mean
   $\frac{n}{n+1}$. Hence, the expected payoff of player $A$ is $\frac{n}{n+1}$.

  By symmetry of player $A$'s payoff functions, all her actions will be played
  with probability $\frac{1}{n}$. Thus, player $B$'s expected payoff is maximized when
  choosing $s_B(s_A^1,\theta_B)$, yielding her an expected payoff of $\frac{\mu_1}{n}$.
  As a result, the expected social welfare in equilibrium is  $\frac{n}{n+1}+\frac{\mu_1}{n}$.

  To compute the optimal social welfare note that if player A select $s_A^1$ with probability $1$,
  her expected payoff is $\frac{1}{2}$ and the expected
  payoff of player $B$ is $\mu_1$ and so the social welfare is $\frac{1}{2}+\mu_1$, which is greater than
   $\frac{n}{n+1}+\frac{\mu_1}{n}$ if $\mu_1\geq \frac{1}{2}$.
  The price of anarchy is  $\left(\frac{1}{2}+\mu_1\right)/\left( \frac{n}{n+1}+\frac{\mu_1}{n}\right)$,
  and in the limit as $n\rightarrow \infty$, the PoA becomes $\frac{1}{2}+\mu_1$.
  Notice that the PoA is an increasing function of $\mu_1$.
\end{example}

\noindent
Examples \ref{ex:1b} and \ref{ex:1} illustrates that, when no transfers are allowed, the price of anarchy can be arbitrarily large as
player $B$ 's payoff is large compare with the payoff of player $A$. We now generalize this idea to quantify the price of anarchy in one-way games.
\begin{proposition}\label{lem:poa}
  In one-way games, the price of anarchy when no payments are allowed
  satisfies, for any type $\theta$,
\[
 %\frac{\max_{s\in \mathcal{S}}u_{B}(s,\theta_B)}{SW(s^N(\theta),\theta)}
\frac{\max_{s\in \mathcal{S}}u_{B}(s,\theta_B)}{\max_{s\in \mathcal{S}}u_{A}(s,\theta_A)+u_B(s^N(\theta),\theta_B)}
 \leq PoA(\theta) \leq 1+\frac{\max_{s\in \mathcal{S}}u_{B}(s,\theta_B)}{\max_{s\in \mathcal{S}}u_{A}(s,\theta_A)},
\]
\end{proposition}
\begin{proof}
  Let $\overline{u}_i(\theta_i)=\max_{s\in
    \mathcal{S}}u_{i}(s,\theta_i)$, $i\in \{A,B\}$.  Independence of
  player $A$ implies that, for all $\theta\in \Theta$, her payoff is
  $u_A(s^N(\theta),\theta_A) = \overline{u}_A(\theta_A)$. It follows that
\begin{eqnarray*}
\frac{\max \{\overline{u}_A(\theta_A), \overline{u}_B(\theta_B)\} }{\overline{u}_A(\theta_A)+u_B(s^N(\theta),\theta_B)} &
 \leq  & PoA(\theta) \\
  \leq  \frac{\overline{u}_A(\theta_A)+\overline{u}_B(\theta_B) }{\overline{u}_A(\theta_A)+u_B(s^N(\theta),\theta_B)}
 & \leq & \frac{\overline{u}_A(\theta_A)+\overline{u}_B(\theta_B)}{\overline{u}_A(\theta_A)}
  = 1+\frac{\overline{u}_B(\theta_B)}{\overline{u}_A(\theta_A)}.
\end{eqnarray*}
\end{proof}

\begin{sloppypar}
\noindent
The price of anarchy can thus be arbitrarily large.  When it is large
enough, Proposition \ref{lem:poa} indicates that
$\max_{s\in \mathcal{S}}u_{B}(s,\theta_B) \geq \max_{s\in
  \mathcal{S}}u_{A}(s,\theta_A) \geq u_B(s^N(\theta),\theta_B)$.
In this case, player $B$ has bargaining power to incentivize player
$A$ monetarily so that she moves from her equilibrium and cooperates
to overcome a bad social welfare. This paper explores this possibility by analyzing
the social welfare when side payments are allowed.
\end{sloppypar}

\section{Related Work}

Before moving to the main results, it is useful to discuss related
games.  One-way games may seem to resemble Stackelberg games with
their notions of leader and follower. The key difference however is
that, in one-way games, the leader does not depend on the action taken
by the follower. In addition, in one-way games, players do not have
complete information and moves are simultaneous. Jackson and Wilkie
(\citeyear{jackson2005endogenous}) studied one-way instances derived
from their more general framework of endogenous games.  However, they
tackled the problem from a different perspective and assumed complete
information (i.e., the player utilities are not private). Jackson and
Wilkie gave a characterization of the outcome when players make
binding offers of side payments, deriving the conditions under which a
new outcome becomes a Nash equilibrium or remains one. They analyzed a
subclass, called 'one sided externality', which is essentially a
one-way game but with complete information. They showed that the
efficient outcome is an equilibrium in this setting, supporting
Coase's claim that a polluter and his victim can reach an efficient
outcome. Under perfect information, the victim can determine the
minimal transfer necessary to support the efficient play.
Naturally, this result does not hold under incomplete information \citep{myerson1983efficient}.
In what follows, we design a bargaining mechanism that is able to cope with the incomplete information setting.

\section{Bayesian-Nash Mechanisms}

In this section, we consider a Bayesian-Nash setting with quasi-linear
preferences. Both players $A$ and $B$ have private utilities and
beliefs about the utilities of the other players. By the revelation
principle, we can restrict our attention to direct mechanisms which
implement a social choice function.  A social choice function in
quasi-linear environments takes the form of $f(\theta)=\left(
  k(\theta), t(\theta) \right)$ where, for every $\theta \in \Theta$,
$k(\theta)\in \mathcal{S}$ is the allocation function and $t_i(\theta)
\in \mathbb{R}$ represents a monetary transfer to agent $i$.  The main
objective of mechanism design is to implement a social choice function
that achieves near efficient allocations while respecting some
desirable properties. For completeness, we specify these key
properties.

\begin{defn}
\label{def-vcg}
A social choice function is ex-post efficient if, for all $\theta \in
\Theta$, we have
$ k(\theta) \in \argmax_{s \in \mathcal{S}}\sum_i u_i(s,\theta)$.
\end{defn}

\begin{defn}
A social choice function is budget-balanced (BB) if, for all $\theta \in
\Theta$, we have
$\sum_i t_i(\theta)=0$.
\end{defn}
\noindent
In other words, there are no net transfers out of the system or into the system.
Taken together, ex-post efficiency and budget-balance imply Pareto optimality.
An essential condition of any mechanism is to guarantee that agents report their true types.
The following property captures this notion when agents have prior beliefs
on the types of other agents.
\begin{defn}
  A social choice function is Bayes-Nash incentive compatible (IC) if for every player $i$:
\[
 \mathbb{E}_{\theta_{-i}|\theta_i} \left[ u_i(k(\theta_i,\theta_{-i}),\theta_i) + t_i(\theta_i,\theta_{-i}) \right] \geq
\mathbb{E}_{\theta_{-i}|\theta_i}  \left[  u_i(k(\hat{\theta}_i,\theta_i),\theta_i) + t_i(\hat{\theta}_i,\theta_{-i}) \right]
\]
where $\theta_i \in \Theta_i$ is the type of player $i$,
$\hat{\theta}_i$ is the type player $i$ reports, and
$\mathbb{E}_{\theta_{-i}|\theta_i} $ denotes player $i$'s expectation
over prior beliefs $\theta_{-i}$ of the types of other agents given
her own type $\theta_i$.
\end{defn}
\noindent
The most natural definition of individual-rationality (IR)
is \emph{interim} IR, which states that every agent type has non-negative
expected gains from participation.

\begin{defn}
  A social choice function is interim individual-rational if, for all types
  $\theta \in \Theta$, it satisfies
\[
 \mathbb{E}_{\theta_{-i}|\theta_i} \left[ u_i(k(\theta),\theta_i) + t_i(\theta) \right] \geq \overline{u}_i(\theta_i),
\]
where $\overline{u}_i(\theta_i)$ is the expected utility for non-participation.
\end{defn}

\noindent
In the context of one-way games, both players have positive outside
options that depend only in their types. In particular, the outside
options are given by the Nash equilibrium outcome under no side
payments.  For players $A$ and $B$, the expected utilities for
non-participation are
$\overline{u}_A(\theta_A)=u_A(s^N_A(\theta_A),\theta_A)$ and
$\overline{u}_B(\theta_B)=u_B(s^N(\theta),\theta_B)$ respectively.

%\section{Budget-Balanced Mechanisms}
%\label{section-balanced}
%In this section, we search for Bayesian-Nash mechanisms without
%subsidies, i.e., budget-balanced mechanisms.

\subsection{Impossibility Result}\label{impossibility_result_section}

This section shows that there exists no mechanism for one-way games
that is efficient and satisfies the traditional desirable
properties. The result is derived from the Myerson-Satterthwaite
(\citeyear{myerson1983efficient}) theorem, a seminal impossibility
result in mechanism design. The Myerson-Satterthwaite theorem considers a
bargaining game with two-sided private information and it states that,
for a bilateral trade setting, there exists no Bayes-Nash
incentive-compatible mechanism that is budget balanced, ex-post
efficient, and gives every agent type non-negative expected gains from
participation (i.e., ex interim individual rationality).

Our contribution is twofold: we present an impossibility
result for one-way games and we relate them with bargaining games, an
idea that we will further explore on the following sections.
We now formalize the impossibility result for one-way games.

Consider the Myerson-Satterthwaite bilateral bargaining setting.

\begin{defn}{Myerson-Satterthwaite bargaining game:}
\begin{enumerate}
  \item  A seller (player 1) owns an object for which her valuation is $v_1 \in V_1$,
  and a buyer (player 2) wants to buy the object at a valuation $v_2\in V_2$.

 \item Each player $i$ knows her valuation
 $v_i$ at the time of the bargaining and player 1 (resp. 2) has a
 probability density distribution $f_2(v_2)$ (resp. $f_1(v_1)$) for the other player's
 valuation.

 \item Both distributions are assumed to be continuous and positive on their domain,
 and the intersection of the domains is not empty.
\end{enumerate}
\end{defn}

By the revelation principle, we can restrict our attention to incentive-compatible
direct mechanisms.  A direct mechanism for bargaining games is
characterized by two functions: (1) a probability distribution
$\sigma: V_1\times V_2 \rightarrow [0,1]$ that specifies the probability that
the object is transferred from the seller to the buyer and (2) a monetary transfer
scheme $p: V_1\times V_2 \rightarrow \mathbb{R}^2$. In this setting, ex-post efficiency
is achieved if $\sigma(v_1,v_2) = 1$ when $v_1<v_2$, and $0$ otherwise.

Our result consists in showing that a mechanism ${\cal M}'$ for the
Myerson-Satterthwaite setting can be constructed using a mechanism
${\cal M}$ for a one-way game in such a way that, if ${\cal M}$ is
 efficient, individual-rational (IR), incentive compatible (IC), and budget-balanced (BB),
 then ${\cal M}'$ is efficient, IR,
IC, and BB. The Myerson-Satterthwaite impossibility theorem
states that such a mechanism ${\cal M}'$ cannot exist, which implies the following
impossibility result for one-way games.

\begin{theorem}
  There is no ex-post efficient, individually rational,
  incentive-compatible, and budget-balanced mechanism for one-way
  games.
\end{theorem}
\begin{proof}

  For any bargaining setting, consider the following transformation into a
  one-way game instance:
\[
\mathcal{S}_A = \{s_A^1, s_A^2\}, \mathcal{S}_B = \{s_B\},
\]
\[
\forall v_1 \in V_1: u_A(s_A^1,v_1) = v_1,\; u_A(s_A^2,v_1) = 0,
\]
\[
\forall v_2 \in V_2: u_B((s_A^1,s_B), v_2) = 0,\;  u_B((s_A^2, s_B),v_2) = v_2,
\]
where player types $(v_1,v_2) \in V_1\times V_2$ are drawn from
distribution $f_1\times f_2$.  Two possible outcomes may occur,
$(s_A^1,s_B)$ or $(s_{A}^2, s_B)$, with social welfare $v_1$ and $v_2$
respectively.

Let us assume ${\cal M}= (k,t)$ is a direct mechanism for one-way
games and that ${\cal M}$ is ex-post efficient, IR, IC, and BB.  We now construct a
mechanism ${\cal M}'= (\sigma,p)$, where $\sigma(v_1,v_2)$ is the
probability that the object is transferred from the seller to the
buyer and $p(v_1,v_2)$ is the payment of each player.  We define
${\cal M}'$ such that
  \[
  \sigma(v_{1},v_{2})=
  \begin{cases}
   0 & \text{if } k(v_1,v_2)=(s_{A}^{1},s_{B}),\\
   1 & \text{if } k(v_1,v_2)=(s_{A}^{2},s_{B}),
  \end{cases}
\]
and \[ p(v_1,v_2)=t(v_1,v_2). \]

   It remains to show that ${\cal M}'$ satisfies all the desired properties.
  An ex-post efficient mechanism ${\cal M}$ in the one-way instance
 satisfies
\[
k(v_{1},v_{2})=
\begin{cases}
(s_{A}^{1},s_{B}) & \text{if }v_{1}\geq v_{2},\\
(s_{A}^{2},s_{B}) & \text{if }v_{1}<v_{2}.
\end{cases}
\]
Therefore, $\sigma(v_{1},v_{2})$ will assign the object to the buyer
iff $v_{1}<v_{2}$.  That is, the player with the highest valuation
will always get the object, meeting the restriction of ex-post
efficiency.  The budget-balanced constraint in ${\cal M}$ implies that
$ p_1(v_1,v_2) + p_2(v_1,v_2) = 0$ for all possible valuations, so
${\cal M}'$ is budget-balanced.

The individual rationality property for ${\cal M}'$ comes from
noticing that the default strategy of player $A$ when no payments are
allowed is $s_A^1$ and the corresponding payoff is $v_1$.  Therefore,
the seller utility is guaranteed to be at least her valuation
$v_1$. Analogously, the buyer will not have a negative utility given
that $u_B((s_A^1,s_B), v_2) = 0$.

Incentive-compatibility is straightforward from definition. Assume
that ${\cal M}'$ is not incentive-compatible, then in mechanism ${\cal
  M}$, at least one player could benefit from reporting a false type.

Such a mechanism ${\cal M}'$ cannot exist since it contradicts
Myerson-Satterthwaite impossibility result, which concludes our proof.
\end{proof}

An immediate consequence of this result is that Bayesian-Nash
mechanisms can only achieve at most two of the three properties:
ex-post efficiency, individual-rationality, and budget balance. For
instance, VCG and dAGVA \citep {d1979incentives,arrow1977property} are
part of the Groves family of mechanisms that truthfully implement
social choice functions that are ex-post efficient. VCG has no
guarantee of budget balance, while dAGVA is not guaranteed to meet the
individual-rationality constraints.  We refer the reader to Williams
(\citeyear{williams1999characterization}) and Krishna and Perry
(\citeyear{krishna1998efficient}) for alternative derivations of the
impossibility result for bilateral trading under the Groves family of
mechanisms.

\section{Single-Offer Mechanism} \label{section_single_offer}

In this section, we propose a simple bargaining mechanism for player
$B$ to increase her payoff. The literature about bargaining games is
extensive and we refer readers to a broad review by Kennan and Wilson
(\citeyear{kennan1993bargaining}).

Given the nature of our applications, individual rationality imposes a
necessary constraint. Otherwise, player $A$ can always defect from
participating in the mechanism and achieve her maximal payoff
independently of the type of player $B$. Additionally, we search for
Bayesian-Nash mechanisms without subsidies, i.e., budget-balanced
mechanisms. The lack of a subsidiary in this case gives rise to a
decentralized mechanism that does not require a third agent to perform
the computations needed by the mechanism. However, a third party is
needed to ensure compliance with the agreement reached by both
players.

 An interesting starting point for
one-way games is the recognition that, whenever player $B$ has a
better payoff than $A$, player $A$ may let player $B$ play her optimal
strategy in exchange for money. The resulting outcome can be viewed as
swapping the roles of both players, i.e., player $B$ chooses her
optimal strategy and $A$ plays her best response to $B$'s strategy. In
this case, as in Proposition \ref{lem:poa}, the worst outcome would be
\[
1+\frac{\max_{s\in \mathcal{S}}u_{A}(s,\theta_A)}{\max_{s\in
    \mathcal{S}}u_{B}(s,\theta_B)}.
\]
This observation leads to the following lemma.

\begin{lemma}\label{lem:poa2}
Consider the social choice function that selects the best strategy
that maximizes the payoff of either player $A$ or player $B$, i.e.,
the strategy
\[
s'(\theta) =\argmax_{s\in \mathcal{S}}
\left( \max \left(u_A(s,\theta_A),u_B(s,\theta_B)\right) \right) .
\]
In the one-way game,
strategy $s'(\theta)$ has a price of anarchy of 2 (i.e., $ \forall \theta \; PoA(\theta)=2$).
\end{lemma}

\noindent
Unfortunately, this social choice function cannot be implemented in
dominant strategies without violating individual rationality. Player
$A$ may have a smaller payoff by following strategy $s'$ instead of
the Nash equilibrium strategy $s^N$.   Indeed, when
$SW(s',\theta) < SW(s^N,\theta)$, it must be that at least one of
the players will be worse than playing the Nash equilibrium strategy
$s^N$.
Lemma \ref{lem:poa2} however
gives us hope for designing a budget-balanced mechanism that has a
constant price of anarchy. Indeed, a simple and distributed
implementation would ask player $B$ to propose an action to be
implemented and player $A$ would receive a monetary compensation for
deviating from her maximal strategy.

We now present such a distributed implementation based on a bargaining
mechanism. The mechanism is inspired by the model of two-person
bargaining under incomplete information presented by Chatterjee and
Samuelson (\citeyear{chatterjee1983bargaining}). In their model, both
the seller and the buyer submit sealed offers and a trade occurs if
there is a gap in the bids. The price is then set to be a convex
combination of the bids. Our single-offer mechanism adapts this idea
to one-way games. In particular, to counteract player $A$'s advantage,
player $B$ makes the first and final offer. Moreover, the structure of
our mechanism makes it possible to quantify the price of anarchy and
provide quality guarantee on the mechanism outcome. Our single-offer
mechanism is defined as follows:

%Single-offer mechanism:
\begin{enumerate}
\item  Player B selects an action $s_A \in \mathcal{S}_A$ to propose to player $A$.
\item Player $B$ also computes her outside option $s_B^O(s_A, \theta_B)$ in case player $A$ rejects action $s_A$,  
and we denote by $u_B^O(s_A, \theta_B)$ the expected payoff from her outside option.
 \item Player $B$ proposes a monetary value of $\gamma \cdot \Delta_B(s_A,\theta_B)$
  with $\Delta_B(s_A,\theta_B)=u_B(s_B(s_A,\theta_B),\theta_B) - u_B^O(s_A, \theta_B)$
     and $\gamma \in \mathbb{R}_{[0,1]}$ to player $A$ in the hope
     that she accepts to play strategy $s_A$ instead of strategy
     $s^N_A$.
\item Player $A$ decides whether to accept the offer.
\item If player $A$ accepts the offer, the outcome of the game is  $\left( s_A, s_B(s_A,\theta_B) \right)$.
Otherwise the outcome of the game is the outside option  $\left(s^N_A(\theta_A), s_B^O(s_A, \theta_B)\right)$.
\end{enumerate}

It is worth observing that a broker is required in this mechanism to
ensure that the outcome $\left(s^N_A(\theta_A), s_B^O(s_A, \theta_B)\right)$
is implemented if player A rejects the unique offer, and no counteroffers are made. A key feature of the
single-offer mechanism is that it requires a minimum amount of information from player A
(i.e., whether she accepts or rejects the offer).

To derive the equilibrium strategy for the single-offer mechanism, we assume players are expected utility maximizers.
The parameter $\gamma \in \mathbb{R}_{[0,1]}$ has been chosen so that
player $B$, satisfying individual rationality, never offers more than $\Delta_B(s_A,\theta_B)$ and her
payoff is never worse than her expected outside option $u_B^O(s_A, \theta_B)$.
Whereas the mechanism can only guarantee interim individual rationality for player $B$, it
provides ex-post individual rationality for player $A$, as shown in the following proposition.

\begin{proposition}\label{prop:accept}
  If players $A$ and $B$ play the single-offer mechanism, for any $(\theta_A,\theta_B) \in \Theta$, player $A$
  accepts the offer $(\gamma,s_A)$ whenever $$u_A(s_A,\theta_A) + \gamma \cdot \Delta_B(s_A,\theta_B) \geq
  u_A(s_A^N(\theta_A),\theta_A).$$	
\end{proposition}

\noindent
In case player $A$ rejects the offer $(\gamma,s_A)$, she will choose her utility maximizing action $s_A^N(\theta_A)$ as
her outside option. Note that by Proposition \ref{prop:accept}, if $s_A=s_A^N(\theta_A)$,  player $A$
would never reject the proposed action $s_A$. Accordingly, if proposed action $s_A$ is rejected then
$s_A\not=s^N_A(\theta_A)$. This observation leads to the following proposition.

\begin{proposition}\label{prop:outside}
For every task $s \in \mathcal{S}_A$, let $\Theta^{s}_A = \Theta_A \setminus \{\theta_A \in \Theta_A: s = \argmax_{x \in \mathcal{S}_A} u_A(x,\theta_A)\}$. In the single-offer mechanism, player $B$ will pick outside option $s^O_B(s_A, \theta_B)$ such that her expected payoff
is maximized, i.e.,
\[
s_B^O(s_A, \theta_B) =  \argmax_{s_B \in \mathcal{S}_B} \mathbb{E}_{\theta^{s_A}_A} \left[ u_B\left((s^N_A(\theta_A),s_B), \theta_B \right)\right].
\]
\end{proposition}

\setcounter{ex}{0}
\begin{example}(continued)
\label{ex:2b}
\normalfont The payoff of Player $B$ is higher if action $s_A^1$ is
played by player $A$.  Hence, player $B$ has incentives to submit an
offer $c$ that triggers action $s_A^1$. Player $A$ accepts the offer if
$c+u_{A}(s_A^1)\ge u_{A}(s_A^2)=100$. Given that $u_{A}(s_A^1)$
follows a uniform distribution, the probability that player $A$
accepts the offer is $\frac{c}{100}$ if $c\leq100$ and $1$ otherwise. For player $B$, this offer has an expected
payoff of $\frac{c}{100}\cdot\left(x-c\right)$ if $c\leq100$ and $x-c$ otherwise. The
optimal value for the offer is given by $c^{*}=\frac{x}{2}$  if $x\leq200$ and
$c^{*}=100$  if $x>200$. This
leads to an expected social welfare for the single-offer
mechanism of
\[
SW=\begin{cases}
100+x\left(\frac{x}{200}-\frac{1}{4}\right) & \mbox{if }x\leq200,\\
50+x & \mbox{if }x>200
\end{cases}
\]
Recall that the optimal social welfare is $50+x$ if
$x\geq50$ and $100$ otherwise. Therefore, the mechanism has a price of
anarchy,
\[
PoA=\begin{cases}
\frac{100}{100+x\left(\frac{x}{200}-\frac{1}{4}\right)} & \mbox{ if }x\leq50,\\
\frac{50+x}{100+x\left(\frac{x}{200}-\frac{1}{4}\right)} & \mbox{ if }50\leq x\leq200,\\
\frac{50+x}{50+x} & \mbox{ if }200\leq x.
\end{cases}
\]
The PoA is bounded by a constant and in fact, $PoA \leq 1.21$ for any $x$.
This contrasts with the unbounded PoA obtained when no side payments are
allowed.
\end{example}
\begin{example}(continued)
\label{ex:2}
\normalfont
The payoff of player $B$ is higher if action $s_A^1$ is
played by player $A$.  Hence, player $B$ has an incentive to submit a monetary
offer $c\leq 1$ that triggers action $s_A^1$. Player $A$ accepts the offer if
$c+u_{A}(s_A^1)\geq \max_{s_A\in\mathcal{S}_A} u_{A}(s_A)$.
It can be shown that the probability that player $A$ accepts the offer is $\frac{c n-c^n}{n-1}$.
In case of acceptance, the expected payoff is $\mu_1-c$ for player $B$
 and $\frac{1}{2}+c$ for player $A$. In case of rejection, it is guaranteed that player A will not play $s_A^1$ and hence
 player $B$'s outside option is action $s_B(s_A^2,\theta_B)$ with an expected payoff of $\frac{\mu_2}{n-1}$.
Player $A$'s expected outside option is $\frac{n}{n+1}$, corresponding to her expected maximum payoff derived in Example \ref{ex:1}.
As a result, player $B$ by offering $c$, has an expected payoff of
 \[
 \frac{c n-c^n}{n-1}\left(\mu_1-c \right)+\left( 1-\frac{c n-c^n}{n-1} \right)\frac{\mu_2}{n+1}.
 \]
 When n is large, the probability of acceptance is approximately,
\[
\lim_{n\rightarrow \infty}\frac{c n-c^n}{n-1}= c.
\]
 Accordingly, in case of rejection, the expected payoffs become $\lim_{n\rightarrow \infty} \frac{\mu_2}{n-1}=0$ for player $B$, and $\lim_{n\rightarrow \infty} \frac{n}{n+1}=1$ for player $A$.
Player $B$'s expected payoff is thus $c \left(\mu_1-c \right)$ and is
maximized when she offers
$c^{*}=\frac{\mu_1}{2}$ if $\mu_1\leq 2$ and $c^{*}=1$ otherwise.
This leads to an expected social welfare for the single-offer mechanism of
$SW= c^*(\mu_1+\frac{1}{2}) + (1-c^*)(0+1)$.
Recall that the optimal social welfare is $\frac{1}{2}+\mu_1$ if $\mu_1 \geq \frac{1}{2}$ and $1$ otherwise.
Therefore the mechanism has the following price of anarchy,
\[
PoA=\begin{cases}
\dfrac{1}{\frac{\mu_{1}}{2}(\mu_{1}+\frac{1}{2})+(1-\frac{\mu_{1}}{2})} & \mbox{ if }\mu_{1}\leq\frac{1}{2},\\
\dfrac{\frac{1}{2}+\mu_{1}}{\frac{\mu_{1}}{2}(\mu_{1}+\frac{1}{2})+(1-\frac{\mu_{1}}{2})} & \mbox{ if }\frac{1}{2}\leq\mu_{1}\leq2,\\
\dfrac{\frac{1}{2}+\mu_{1}}{\frac{1}{2}+\mu_{1}}=1 & \mbox{ if }2\leq\mu_{1},
\end{cases}
\]
and the PoA has a maximum value of $\frac{4}{31} \left(3+2 \sqrt{10}\right) \approx 1.203$.
This contrasts with the unbounded PoA obtained by the Nash equilibrium when no side payments are
allowed.
\end{example}

We now generalize the analysis done in Examples \ref{ex:2b} and \ref{ex:2}.
We proceed by studying the utility-maximizing strategy $(s_A , \gamma)$ for
player $B$ and then derive the expected social welfare of the outcome
for the single-offer mechanism.
 Note that, in case of agreement, the
action of player $B$ of type $\theta_B$ is solely defined by $s_A$ as
she has no incentives to defect from its best response
$s_B(s_A,\theta_B)$.  By Proposition \ref{prop:accept}, player $A$
accepts an offer whenever $ \Delta_A(s_A,\theta_A)\leq \gamma
\Delta_B(s_A,\theta_B)$, where $\Delta_A(s_A,\theta_A) =
u_A(s_A^N(\theta_A),\theta_A)-u_A(s_A,\theta_A)$.  Player B obviously
aims at choosing $\gamma$ and $s_A$ to maximize her payoff and we now
study this optimization problem.  In the case of an agreement, player
$B$ is left with a profit of
$$u_B(s_B(s_A,\theta_B),\theta_B)-\gamma \cdot \Delta_B(s_A,\theta_B).$$
Otherwise, player $B$ gets an expected payoff of $u_B^O(s_A, \theta_B)$.

\begin{defn}\label{def:dist}
The probability that player $A$ accepts the offer $(s_A , \gamma)$, given that player $B$ has type $\theta_B \in \Theta_B$, is
\[
P(s_A,\gamma,\theta_B) = \Pr\left[ \gamma \cdot \Delta_B(s_A,\theta_B) \geq \Delta_A(s_A,\theta_A) \right]
= \int_{\theta_A \in \Theta_A} f_A(\theta_A)\cdot \delta(s_A, \gamma \cdot \Delta_B(s_A,\theta_B), \theta_A) d\theta_A,
\]
with
\[
\delta(s_A,x, \theta_A)=
\begin{cases}
1 & \text{ if } x \geq \Delta_A(s_A,\theta_A),\\
0 & \text{ otherwise. }
\end{cases}
\]
\end{defn}
\noindent
The expected profit of players $A$ and $B$ for proposed action
  $s=(s_A,s_B)$ and $\gamma$ when player $B$ has type $\theta_B$ is
  given by
\begin{eqnarray*}
\mathbb{E}_{\theta_A} \left[U_B(s_A,\gamma,\theta_B)\right] & = & u_B^O(s_A, \theta_B) +
 P(s_A,\gamma,\theta_B) \left( (1 - \gamma) \cdot \Delta_B(s_A,\theta_B) \right),\\
\mathbb{E}_{\theta_A} \left[U_A(s_A,\gamma,\theta_B)\right] & = & \mathbb{E}_{\theta_A} [ u_A^N(\theta_A)] +  P(s_A,\gamma,\theta_B)\cdot
 \left( \gamma \cdot \Delta_B(s_A,\theta_B)- \mathbb{E}_{\theta_A} \left[\Delta_A(s_A,\theta_A) \right] \right).
\end{eqnarray*}

\noindent
The optimal strategy of player $B$ is specified in the following
lemma.

\begin{lemma}\label{lem:opt-gamma}
On the single-offer mechanism, player $B$ chooses $s_A^*(\theta_B)$ and $\gamma^*(s_A^*,\theta_B)$ such that
\[
s_A^*(\theta_B) =  \argmax_{s_A  \in \mathcal{S}_A}
\mathbb{E}_{\theta_A} \left[U_B(s_A,\gamma^*,\theta_B)\right] ,
\]
where
\[
 \gamma^*(s_A,\theta_B) = \argmax_{\gamma \in \mathbb{R}_{[0,1]}} P(s_A,\gamma,\theta_B)  \cdot  (1 - \gamma).
\]
\end{lemma}
%\noindent
%Note that such strategy satisfies the Bayesian incentive compatibility constraints by construction.

\subsection{Price of Anarchy}
 \label{section_price_of_anarchy}

%\begin{defn} \label{def:sw}
%The expected social welfare, for proposed action $s=(s_A,s_B)$
%  and $\gamma$, when player $B$ has type $\theta_B$ and prior beliefs
%  $f_A$ over the types of player $A$, is given by
%\begin{multline*}
%\mathbb{E}_{\theta_A}[SW(s,\gamma,\theta_B)] = \\ (1- P\left(s_A,\gamma,\theta_B \right)) \cdot u_B^N(\theta_B) +\mathbb{E}_{\theta_A}[u_A(s_A^N(\theta_A),\theta_A)]\\
%+  P\left(s_A,\gamma,\theta_B \right) \cdot \left(u_B(s,\theta_B) + \mathbb{E}_{\theta_A}[u_A(s_A,\theta_A)] \right).
%\end{multline*}
%\end{defn}

We now analyze the quality of the outcomes in the single-offer
mechanism.
The first step is the derivation of a lower bound for the expected
social welfare of the single-offer mechanism. Inspired by Lemma
\ref{lem:poa2}, instead of considering all pairs $\langle s_A, \gamma
\rangle$, the analysis restricts attention to a single action $s'_A =
\argmax_{s_A \in \mathcal{S}_A} u_B(s_B(s_A,\theta_B),\theta_B)$.  We
prove that, when offering to player $A$ action $s'_A$ and its
associated optimal value for $\gamma$, the expected social welfare is
lower than the optimal pair $\langle s_A^*, \gamma^* \rangle$. As a
result, we obtain an upper bound to the price of anarchy of the
single-offer mechanism.

To make the discussion precise, consider the strategy where player B offers
$\langle s_A', \gamma^*(s'_A,\theta_B) \rangle$,
with  $\gamma^*(s'_A,\theta_B)$ being the optimal choice of $\gamma$ given
$s'_A$, following the notation used in Lemma \ref{lem:opt-gamma}.

\begin{lemma}\label{lem:upper_bound}
For any type $\theta_B \in \Theta_B$ of player $B$, the expected social welfare achieved by the single-offer mechanism
  is at least the expected social welfare achieved by the strategy
$\langle s_A', \gamma^*(s'_A,\theta_B) \rangle$.
\end{lemma}
\begin{proof}
  Let $\gamma^*= \gamma^*(s_A^*,\theta_B)$ and $\gamma'=
  \gamma^*(s_A',\theta_B)$. The optimality condition of $s^*$ implies that
\begin{equation} \label{eq:opt}
\mathbb{E}_{\theta_A} \left[U_B(s',\gamma',\theta_B)\right] \leq\mathbb{E}_{\theta_A} \left[U_B(s^*,\gamma^*,\theta_B)\right].
\end{equation}
Two cases can occur. The first case is
\[
P(s_A',\gamma',\theta_B) \leq
P(s_A^*,\gamma^*,\theta_B),
\]
i.e., the probability of player
$A$ accepting offer $(s_A^*,\gamma^*)$ is greater than if offered
$(s_A',\gamma')$. Then, it must be that the expected payoff of player
$A$ is greater when offered $(s_A^*,\gamma^*)$, i.e.,
\[
\mathbb{E}_{\theta_A} \left[U_A(s_A',\gamma',\theta_B)\right] \leq\mathbb{E}_{\theta_A} \left[U_A(s_A^*,\gamma^*,\theta_B)\right].
\]
This, together with Inequality (\ref{eq:opt}) results in the
single-offer mechanism having a greater expected social welfare.

  The second case is
  \[
  P(s_A',\gamma',\theta_B) >
  P(s_A^*,\gamma^*,\theta_B).
  \]
  Consider $\gamma''$ such
  that $P(s_A',\gamma'',\theta_B) =
  P(s_A^*,\gamma^*,\theta_B)$. The fact that the
  probabilities of acceptance are the same implies that the expected
  payoff of Player $A$ is the same in both cases, i.e.,
  $\mathbb{E}_{\theta_A} \left[U_A(s_A',\gamma'',\theta_B)\right] =
  \mathbb{E}_{\theta_A}
  \left[U_A(s_A^*,\gamma^*,\theta_B)\right]$. This, together with
  Equation (\ref{eq:opt}) yields
  \[
  \mathbb{E}_{\theta_A}[SW(s^*,\gamma^*,\theta_B)]\geq
  \mathbb{E}_{\theta_A}[SW(s',\gamma'',\theta_B)].
  \]
   This is equivalent to
\begin{equation}\label{eq:sw}
u_B(s^*,\theta_B) + \mathbb{E}_{\theta_A}[u_A(s_A^*,\theta_A)] \geq
u_B(s',\theta_B) + \mathbb{E}_{\theta_A}[u_A(s_A',\theta_A)].
\end{equation}
Similarly, consider $\gamma^{**}$ such that
\[
P(s_A',\gamma',\theta_B) =
P(s_A^*,\gamma^{**},\theta_B),
\] which implies
\[
\mathbb{E}_{\theta_A} \left[U_A(s_A',\gamma',\theta_B)\right] =
\mathbb{E}_{\theta_A}
\left[U_A(s_A^*,\gamma^{**},\theta_B)\right].
\]
Existence of
$\gamma^{**}$ is guaranteed by Inequality (\ref{eq:sw}) which states
that, there is more money in expectation to transfer to player $A$
when choosing $s^*$ over $s'$. The fact that the acceptance
probabilities are the same, together with Inequality (\ref{eq:sw}),
implies that
\[
\mathbb{E}_{\theta_A}[SW(s^*,\gamma^{**},\theta_B)]\geq
\mathbb{E}_{\theta_A}[SW(s',\gamma',\theta_B)].
\]
Given that the
expected payoff of player $A$ is the same in both cases, it must be
the case that the expected payoff of player $B$ is higher when using
$(s_A^*,\gamma^{**})$.

Therefore, we have found an offer for the single-offer mechanism
with greater expected social welfare and a greater payoff for
player $B$ compared with strategy $\langle s_A', \gamma' \rangle$.
\end{proof}

%Note that the proof does not use the fact that $s'_A = \argmax_{s_A
%  \in \mathcal{S}_A} u_B(s_B(s_A,\theta_B),\theta_B)$. Therefore, by
%exchanging $s_A'$ for any $s \in \mathcal{S}_A$, the proof continues
%to hold. This leads to the following corollary.
%\begin{corollary}
%In the single-offer mechanism, the optimal strategy of player $B$ maximizes the expected social welfare.
%\end{corollary}

\noindent
We are ready to derive an upper bound for the induced \emph{price of
  anarchy} for the single-offer mechanism.  We first derive the price
of anarchy of strategy $\langle s_A', \gamma' \rangle$ in case of agreement
and disagreement of player $A$.

\begin{lemma}
\label{lem:single}
Consider action $s'=\arg \max_{s \in \mathcal{S}} u_B(s,\theta_B)$ and
let $PoA^A(\gamma)$ and $PoA^R(\gamma)$ denote the induced price of
anarchy if player $A$ accepts and rejects the offer given a proposed
$\gamma$. Then,
\[
PoA^A(\gamma) = 1+\gamma \quad \text{and} \quad PoA^R(\gamma) = 1 + \frac{1}{\gamma}.
\]
\end{lemma}
\begin{proof}
  Let  $s_A^N=s^N_A(\theta_A)$, $u_A^N = u_A(s_A^N,\theta_A)$, $u_A'=u_A(s',\theta_A)$,
  $u_B'=u_B(s',\theta_B)$, $s_B^O=s^O_B(s_A',\theta_A)$ and  $u_B^O=u_B^O(s_A',\theta_B)$.
   Player $B$ offers action $s_A'$ and a monetary value of $\gamma \Delta_B(s_A')= \gamma (u_B'-u_B^O)$ to player $A$.
    Two cases can occur.

\begin{description}
\item [{Player A accepts:}]
 $u_A' + \gamma \Delta_B(s_A') \geq  u_A^N$.
Strategy $(s_A',s_B')$ is played.
\begin{eqnarray*}
PoA^{A} & \leq & \frac{u_{A}^{N}+u_{B}'}{u'_{A}+u_{B}'}\leq\frac{u'_{A}+u_{B}'+\gamma\cdot u_{B}'}{u'_{A}+u_{B}'}\\
 & = & 1+\gamma\frac{u_{B}'}{u'_{A}+u_{B}'}\leq 1+\gamma.
\end{eqnarray*}
\item [{Player A rejects:}]
$u_A' +  \gamma \Delta_B(s_A')  <  u_A^N$.
Strategy $(s_A^N, s_B^O)$ is played.\\
\[
PoA^{R}  \leq  \frac{u_{A}^{N}+u_{B}'}{u_{A}^N+u_B^O}\leq 1+\frac{u_{B}'}{u_{A}^N+u_B^O}  \leq   1+\frac{1}{\gamma},
 \]
 where the last inequality comes from
 \[
 u_A' +  \gamma \Delta_B(s_A')  <  u_A^N \Leftrightarrow   \gamma u_B'  <  u_A^N +\gamma  u_B^O - u_A' <  u_A^N +  u_B^O.
 \]
\end{description}
\end{proof}

\noindent
When $\gamma = 1$, the price of anarchy is 2 but player $B$
has no incentive to choose such a value. If $\gamma = 0.5$, the price
of anarchy is $3$. Of course, player $B$ will choose $\gamma'=
\gamma^*(s_A',\theta_B)$. Lemma \ref{lem:single} indicates that the
worst-case outcome is $(1+\gamma')$ when player $A$ accepts with a
probability $P(s_A',\gamma',\theta_B)$ and
$(1+\frac{1}{\gamma'})$ otherwise. This yields the following result.

\begin{theorem}\label{th:BayesPoA}
The Bayesian price of anarchy of the single-offer mechanism for one-way games is at most
\[
\frac{\gamma'+1}{\gamma'} \left( 1- P(s_A',\gamma',\theta_B)(1-\gamma') \right),
\]
where
\[
\gamma' = \argmax_{\gamma \in \mathbb{R}_{[0,1]}} P(s_A',\gamma,\theta_B) (1 - \gamma).
\]
\end{theorem}
\begin{proof}
By combining Lemmas \ref{lem:upper_bound} and \ref{lem:single}, we can derive the following upper bound for the PoA.
\begin{eqnarray*}
PoA & \leq & P(s_{A}',\gamma',\theta_{B})PoA^{A}(\gamma')+\left(1-P(s_{A}',\gamma',\theta_{B})\right)PoA^{R}(\gamma')\\
 & = & P(s_{A}',\gamma',\theta_{B})\left(1+\gamma'\right)+\left(1-P(s_{A}',\gamma',\theta_{B})\right)\left(1+\frac{1}{\gamma'}\right)\\
 & = & 1+\frac{1}{\gamma'}+P(s_{A}',\gamma',\theta_{B})\left(\gamma'-\frac{1}{\gamma'}\right)\\
 & = & \frac{\gamma'+1}{\gamma'}+P(s_{A}',\gamma',\theta_{B})\left(\frac{\gamma'^{2}-1}{\gamma'}\right)\\
 & = & \frac{\gamma'+1}{\gamma'}\left(1-P(s_{A}',\gamma',\theta_{B})\left(1-\gamma'\right)\right)
\end{eqnarray*}
\end{proof}

\noindent
To get a better idea of how the mechanism improves the social welfare,
it is useful to quantify the price of anarchy in Theorem
\ref{th:BayesPoA} for a specific class of distributions.

\begin{corollary}
  If $\Delta_A(s_A',\theta_A)$ has a cumulative distribution function
  $F(x)=(x/\Delta_B)^\beta$ between $0$ and $\Delta_B$, with $0 < \beta
  \leq 1$, then $\gamma =\frac{\beta}{\beta+1}$ and the price of
  anarchy is at most
 $$(2+\frac{1}{\beta}) (1-\beta^\beta (1+\beta)^{-(\beta+1)}).$$
 For example, if $\beta=1$, then $F(x)$ is the uniform distribution, $\gamma =
 \frac{1}{2}$, and the expected price of anarchy is at most $2.25$.
\end{corollary}

\noindent
This corollary, in conjunction with Lemma \ref{lem:poa2}, gives us the
cost of enforcing individual rationality, moving from a price of
anarchy of 2 to a price of 2.25 in the case of a uniform
distribution.

The strategy $\langle s_A', \gamma' \rangle$ is of independent interest.
It indicates how a player with
limited computational power can achieve an outcome that satisfies
individual rationality without optimizing over all strategies.

\section{Multi-Offer Mechanism}\label{section_multi_offer}

\label{section-multi}
This section extends the single-offer mechanism by allowing player $B$ to make multiple monetary offers
for the same proposed action.
Our main result shows that making counteroffers under commitment does not improve efficiency over the single-offer
mechanism. By commitment we mean that player $B$ must be able to guarantee that the price schedule she originally
announces will not be modified in the future.

The single-offer mechanism was characterized by a an action $s_A \in \mathcal{S}_A$
and a single value $\gamma \in \mathbb{R}_{[0,1]}$. The multi-offer mechanism is
characterized by a 4-tuple
\[
(s_A \in \mathcal{S}_A, n,\gamma=(\gamma_1,\hdots,\gamma_n) \in \mathbb{R}_{[0,1]}^n,
p=(p_1,\hdots,p_{n}) \in \mathbb{R}_{[0,1]}^{n}),
\]
where $n$ is the number of offers, $(\gamma_1,\hdots,\gamma_n)$ is a
sequence of numbers in $\mathbb{R}_{[0,1]}^n$ to compute the ratios of
$\Delta_B(s_A,\theta_B)=u_B(s_B(s_A,\theta_B),\theta_B) - u_B^O(s_A, \theta_B)$
to be offered, and
$(p_1,\hdots,p_{n})$ is a sequence of probabilities for continuing to
make offers where we assume that $p_1=1$.
The multi-offer mechanism is defined as follows:
\begin{enumerate}
\item  Player B selects an action $s_A \in \mathcal{S}_A$ to propose to player $A$.
\item Player $B$ also computes her outside option $s_B^O(s_A, \theta_B)$ in case player $A$ rejects action $s_A$,  
and we denote by $u_B^O(s_A, \theta_B)$ the expected payoff from her outside option.
\item Player $B$ selects  $\gamma=(\gamma_1,\hdots,\gamma_n) \in \mathbb{R}_{[0,1]}^n$ and
 $p=(p_1,\hdots,p_{n}) \in \mathbb{R}_{[0,1]}^{n}$ , with $p_1=1$. Player $B$ has to commit to this sequence of values
 (in spite of what she learns from player $A$'s actions).
 \item At step $1\leq i \leq n$, player $B$ proposes a monetary value of $\gamma_i \cdot \Delta_B(s_A,\theta_B)$
  with $\Delta_B(s_A,\theta_B)=u_B(s_B(s_A,\theta_B),\theta_B) - u_B^O(s_A, \theta_B)$
     and $\gamma_i \in \mathbb{R}_{[0,1]}$ to player $A$ in the hope
     that she accepts to play strategy $s_A$ instead of strategy
     $s^N_A$.
\item Player $A$ decides whether to accept the offer.
\item If player $A$ accepts the offer, the outcome of the game is  $\left( s_A, s_B(s_A,\theta_B) \right)$.
\item If player $A$ rejects the offer,  set $i\leftarrow i+1$ and go to step 4 with probability $p_i$.
\item Otherwise the outcome of the game is the outside option  $\left(s^N_A(\theta_A), s_B^O(s_A, \theta_B)\right)$.
\end{enumerate}

For ease of notation, we denote $\Delta_B(s_A) = \Delta_B(s_A,\theta_B)$ and
$\Delta_A(s_A) = \Delta_A(s_A,\theta_A)$ for the rest of this section, where $\Delta_A(s_A,\theta_A)= u_A(s_A^N(\theta_A),\theta_A)-u_A(s_A,\theta_A)$.
In the multiple-offer mechanism, player $B$ makes a sequence of offers
$\gamma_i \Delta_B(s_A)$ to player $A$ to play strategy $s_A$. The first
offer is $\gamma_1 \Delta_B(s_A)$. If player $A$ refuses the offer, then
player $B$ makes a second offer $\gamma_2 \Delta_B(s_A)$ with probability
$p_2$. Hence, with probability $1-p_2$, player $B$ makes no offer and
the outcome of the game is  $\left(s^N_A(\theta_A), s_B^O(s_A, \theta_B)\right)$.
In general, at iteration $i$, player $B$ makes an offer $\gamma_i \Delta_B(s_A)$
with probability $p_{i}$ and the outside option is played with probability $1 - p_{i}$.
The mechanism stops when player $A$ accepts an offer or when Player $B$ stops
making offers to player $A$. In this last case, once again, the outside option is played.

Observe that player $A$ could reject an offer even if it is more
profitable than playing her maximizing utility action $s^N_A(\theta_A)$ because she may expect a
better offer in the future. To avoid this behavior, the multi-offer
mechanism imposes a condition on the $\gamma_i$'s and $p_i$'s to
ensure that player A accepts the first offer that gives her a higher
payoff than her default action $s^N_A(\theta_A)$.
Two conditions must hold for player $A$ to accept an offer in step $i\in [1,\ldots,n]$: \\
{\em (a) Individual Rationality}:
\begin{eqnarray}
  \gamma_i \Delta_B(s_A) \geq \Delta_A(s_A),
\label{eq:accept1}
  \end{eqnarray}
  which is equivalent to Proposition \ref{prop:accept}.\\
%\item
{\em (b) Greater expected utility in step $i$ than in step $i+1$}:
\[
  \gamma_i  \Delta_B(s_A)+ u_A(s_A,\theta_A) \geq
 p_{i+1} \left(\gamma_{i+1} \Delta_B(s_A) + u_A(s_A,\theta_A) \right) + (1-p_{i+1})u_A(s_A^N\left(\theta_A),\theta_A\right)
\]
which is equivalent to
\begin{eqnarray}
\frac{\gamma_i - p_{i+1} \gamma_{i+1} }{1 - p_{i+1}}\Delta_B(s_A) \geq \Delta_A(s_A).
\label{eq:accept2}
\end{eqnarray}

\noindent
We now show that the multiple-offer mechanism is in fact equivalent to
the single-offer mechanism.  We use the notation
\[S_i=
\begin{cases}
0 & i=0,\\
\frac{\gamma_{i}-p_{i+1}\gamma_{i+1}}{1-p_{i+1}} & n>i>0,\\
\gamma_n & i=n.
\end{cases}
\]
so that Condition (\ref{eq:accept2}) can be expressed as $S_i  \Delta_B(s_A) \geq \Delta_A(s_A)$.

Note that if player $A$ refuses an offer with $\gamma_i$,
she will also refuse offers with smaller ratios. This observation
leads to the following proposition.
\begin{proposition}\label{prop:increasing} In the multi-offer mechanism,
\[
\gamma_{i+1} > \gamma_i,\; \forall i\in [1,\ldots,n-1].
\]
\end{proposition}
\noindent
Therefore, Proposition \ref{prop:increasing} states that counteroffers should be increasing on time.
\begin{lemma}\label{lem:gamma}
In the multiple-offer mechanism, for all $i\in [1,\ldots,n]$,
\[
  \gamma_i \geq S_i.
\]
\end{lemma}
\begin{proof}
  Assume that $\gamma_{i}<S_{i}$. By definition of $S_i$, it follows
  that $\gamma_{i}-p_{i+1}\gamma_{i}<\gamma_{i}-p_{i+1}\gamma_{i+1}$
  and hence $\gamma_{i}>\gamma_{i+1}$. This contradicts Proposition \ref{prop:increasing},
  stating that the $\gamma$'s are defined as a non-decreasing sequence.
\end{proof}
\begin{corollary}
\label{col:accept1}
Condition (\ref{eq:accept2}) implies Condition (\ref{eq:accept1}).
\end{corollary}
\begin{proof}
$S_i \leq \gamma_i$ and $\Delta_A(s_A) \leq S_i\Delta_B(s_A)$ implies
$\Delta_A(s_A) \leq S_i\Delta_B(s_A) \leq \gamma_i\Delta_B(s_A)$.
\end{proof}

\noindent
If player $A$ rejected the offer in step $i-1$, then Conditions
(\ref{eq:accept1}) and (\ref{eq:accept2}) were both not satisfied in
step $i-1$.  By Lemma \ref{lem:gamma}, only two cases may occur:
\begin{enumerate}
\item
If $\gamma_i\Delta_B(s_A)<\Delta_A(s_A)$ then it must be the case that
\begin{equation}\label{eq:reject1}
S_{i-1} \Delta_B(s_A) \leq \gamma_{i-1} \Delta_B(s_A) <\Delta_A(s_A).
\end{equation}
\item
If $\gamma_i\Delta_B(s_A) \geq \Delta_A(s_A)$ then
\begin{equation}\label{eq:reject2}
S_{i-1}\Delta_B(s_A) <\Delta_A(s_A) \leq \gamma_{i-1}\Delta_B(s_A).
\end{equation}
\end{enumerate}

\noindent
The disjunction of Conditions $(\ref{eq:reject1})$ and $(\ref{eq:reject2})$ yields the following inequality
\begin{equation}\label{eq:reject}
S_{i-1}\Delta_B(s_A) < \Delta_A(s_A).
\end{equation}
By Corollary \ref{col:accept1}, if player $A$ accepts in step $i$ given that she rejected in step $i-1$, we have that
\begin{equation}\label{eq:accept3}
\Delta_A(s_A) \leq S_i \Delta_B(s_A).
\end{equation}
Recalling Definition \ref{def:dist}, the cumulative distribution function of random variable $\Delta_A(s_A)$ was denoted by,
\[
P(s_A,\gamma) = \Pr\left[ \Delta_A(s_A) \leq \gamma \Delta_B(s_A) \right].
\]
Hence the probability of acceptance in step $i$ can be derived from Conditions
(\ref{eq:reject}) and (\ref{eq:accept3}) as
\[
\Pr[S_{i-1}\Delta_B(s_A) <\Delta_A(s_A) \leq S_i \Delta_B(s_A)]= P(s_A,S_i) -  P(s_A,S_{i-1}).
\]
Player $B$ aims at choosing the $\gamma_i$'s, the probabilities $p_i$'s and action $s_A$ to maximize her expected utility,
which is equivalent to the following optimization problem.
 \begin{alignat}{2}
     \max_{\gamma,p} & \sum^n_{i=1}\left[ \left( \prod^i_{j=1}p_j \right) \left( P(s_A,S_i) -  P(s_A,S_{i-1}) ) \right) (1-\gamma_i)\right]  \label{eq:max1} \\
    \text{s.t. } 		& p_1 = 1,\\
    				& S_0 = 0, \\
    				& S_i = \frac{\gamma_{i}-p_{i+1}\gamma_{i+1}}{1-p_{i+1}}, \label{eq:max2}  \\
    				& S_n = \gamma_n, \label{eq:max3}  \\
    				& S_1\leq S_2 \leq \ldots \leq S_n\leq 1.
  \end{alignat}
Where the term $\prod^i_{j=1}p_j $ is the probability of reaching to the $i$-th offer.
We are now ready to state the main result of this section.
\begin{theorem} The multi-offer mechanism is equivalent to the single-offer mechanism
in one-way games.
\end{theorem}
\begin{proof}
\noindent
By Equation (\ref{eq:max2}),
\[
 (1-p_{i+1}) (1-S_i) = (1 - \gamma_{i}) - p_{i+1} (1 - \gamma_{i+1}).
\]
Then, by using (\ref{eq:max3}) and grouping the $ P(s_A,S_i) $ terms, the objective function
becomes
 \begin{equation}\label{eq:obj}
  \sum^{n-1}_{i=1}\left[  \left( \prod^i_{j=1}p_j \right)  (1-p_{i+1}) P(s_A,S_i) (1-S_i)  \right] +   \left( \prod^n_{j=1}p_j \right) P(s_A,\gamma_n) (1-\gamma_n).
\end{equation}
Observe that each term in the objective function features an
expression of the form $ P(s_A,x)(1-x)$. Hence the objective is
bounded by above by
\[
   (\ref{eq:obj}) \leq \sum^{n-1}_{i=1} \left[ \left( \prod^i_{j=1}p_j \right)
     (1-p_{i+1}) \cdot C \right] + C \cdot \prod^n_{j=1}p_j
\]
where $C = \max_x  P(s_A,x)(1-x)$. We show that, for any given
probabilities $p$, there is a unique solution that meets this upper
bound. Let $x^* = \arg\max_x P(s_A,x)$. The right-hand term
in (\ref{eq:obj}) is optimized by setting $\gamma_n = x^*$. We show by
induction that all the other terms are optimized by setting
$\gamma_i=x^*$. Assume that this holds for
$\gamma_{i+1},\ldots,\gamma_n$. We need to optimize $P(s_A,S_i)(1-S_i)$. By induction,
\[
S_i = \frac{\gamma_{i}-p_{i+1}x^*}{1-p_{i+1}}
\]
and assigning $x^*$ to $\gamma_i$ gives $S_i = x^*$ and $P(s_A,S_i)(1-S_i) = C$.
Since all $\gamma_i$ are equal, this concludes
the proof.
\end{proof}

\noindent
The above derivation is related to a well-known result from Sobel and
Takahashi (\citeyear{sobel1983multistage}), which models an iterative
bargaining where there is a buyer with a private reservation
price and a seller with reservation price 0 who makes all the offers.
There is a known fixed discount factor for each player and, when these
discount factors are equal (this is equivalent to have a probability
for a next offer), they showed that, under commitment, the infinite
horizon bargaining is equivalent to the single shot. There are
differences between their model and ours: In our model, the buyer is
making the offers, the probabilities are not fixed a priori (Player
$B$ can choose them), and both outside options are private.

\section{Discussion}
\label{section-conclusion}

In one-way games, the utility of one player does not depend on the
decisions of the other player. We showed that, in this setting, the
outcome of a Nash equilibrium can be arbitrarily far from the social
welfare solution. We also proved that it is impossible to design a
Bayes-Nash incentive-compatible mechanism for one-way games that is
budget-balanced, individually rational, and efficient. To alleviate
these negative results, we proposed two privacy-preserving mechanisms:
a single-offer and a multi-offer mechanism and showed that both are equivalent.

The single-offer mechanism is simple for both parties, as well as for
the broker who just makes sure that the players follow the
protocol. This mechanism also requires minimal information from the
agents who perform all the combinatorial computations, while it
incentivizes them to cooperate towards the social welfare in a
distributed setting. Moreover, the mechanism has the following
desirable properties: it is budget-balanced and satisfies the
individual rationality constraints and Bayesian
incentive-compatibility conditions. Additionally, we showed that, in a
realistic setting, where agents have limited computational resources,
a simpler version of the mechanism can be implemented without overly
deteriorating the social welfare.

It is an open question whether there exists another mechanism
(possibly more complex) that could lead to a better efficiency, while
keeping the above properties. Indeed, in one-way games, player $A$ has
a intrinsic advantage over player $B$, which is not easy to overcome.
One possible promising mechanism consists of player $B$ setting rewards
for all player $A$'s actions, and player $A$ choosing one in return for that money.
This is known as the Bayesian Unit-demand Item-Pricing Problem (BUPP) \citep{chawla2007algorithmic}.
Recent work has shown this problem to be NP-hard \citep{Chen2014}, but a factor 3 approximation to
the optimal expected revenue of player $B$ is obtained in \citep{chawla2007algorithmic} (subsequently improved
to 2 in \citep{Chawla2010}). In the context of our paper, several interesting questions arise from
the Bayesian Unit-demand Item-Pricing Problem. What is the efficiency achieved by the BUPP in one-way games? What is the impact of a constant factor approximation for the revenue on the social welfare in one-way games?

There are also many other directions for future research. It is
important to generalize one-way games to multiple players. Moreover,
there are applications where the dependencies are in both directions,
e.g., the restoration of the power and the gas systems considered in
Coffrin et al. (\citeyear{coffrin2012last}). These applications
typically have multiple components to restore and the dependencies
form an acyclic graph. Hence such a mechanism would likely need to
consider this internal structure to obtain efficient outcomes.

\section*{Acknowledgments}

NICTA is funded by the Australian Government through the Department of Communications and the Australian Research Council through the ICT Centre of Excellence Program.

\bibliographystyle{plain}
\bibliography{One-way}

\end{document}